\documentclass[12pt]{article}
\usepackage{graphicx,multicol}
\usepackage{soul}
\usepackage{algorithm}
\usepackage{algpseudocode}

\usepackage[a4paper, left=0.7in, right=0.7in]{geometry} 

\usepackage[latin1]{inputenc}

\usepackage{color}
\usepackage{amsthm,amsmath}
\usepackage{amssymb}
\usepackage{amsfonts}
\usepackage{hyperref}
\hypersetup{colorlinks,allcolors=black}
\usepackage{graphicx}
\usepackage{tabularx}
\usepackage{lscape}
\usepackage{pdflscape}

\usepackage{adjustbox} 
\usepackage{array} 

\newtheorem{theorem}{Theorem}
\newtheorem{definition}{Definition}
\newtheorem{corollary}{Corollary}
\newtheorem{proposition}{Proposition}
\newtheorem{lemma}{Lemma}
\newtheorem{example}{Example}
\newtheorem{remark}{Remark}

\date{ }

\author{Atif Ahmad Khan$^{1,}\footnote{Corresponding author }$,   Shakir Ali$^{2}$  , and Bhupendra Singh$^3$ \\
		\small{$^1$Department of Mathematics, Faculty of Science, }\\
	\small{Aligarh Muslim University, Aligarh-202002, India }\\
	\small{$^2$Department of Mathematics, Faculty of Science, }\\
	\small{Aligarh Muslim University, Aligarh-202002, India and }\\
	\small{Faculty of Mathematics and Natural Sciences, Universitas Gadjah Mada, Yogyakarta-55281, Indonesia}\\
	\small{$^3$Centre for Artificial Intelligence \& Robotics, DRDO, Bengaluru, Karnataka-560093, India }\\
	\small{atifkhanalig1997@gmail.com, shakir.ali.mm@amu.ac.in,  bhusinghdrdo@gmail.com}}

\begin{document}
\title{\textbf{On the generalization of $g$-circulant MDS matrices}}
\maketitle
	\begin{abstract}
A matrix $M$ over the finite field $ \mathbb{F}_q $ is called \emph{maximum distance separable} (MDS) if all of its square submatrices are non-singular. These MDS matrices are very important in cryptography and coding theory because they provide strong data protection and help spread information efficiently.  In this paper, we introduce a new type of matrix called a \emph{consta-$g$-circulant matrix}, which extends the idea of $g$-circulant matrices. These matrices come from a linear transformation defined by the polynomial 
	$
	h(x) = x^m - \lambda + \sum_{i=0}^{m-1} h_i x^i
	$
	over $ \mathbb{F}_q $. We find the upper bound of  such matrices exist and give conditions to check when they are invertible. This helps us know when they are MDS matrices. If the polynomial $ x^m - \lambda $ factors as 
	$
	x^m - \lambda = \prod_{i=1}^{t} f_i(x)^{e_i},
	$
	where each \( f_i(x) \) is irreducible, then the number of invertible consta-$g$-circulant matrices is
	$
	N \cdot \prod_{i=1}^{t} \left( q^{\deg f_i} - 1 \right),
	$
	where $r$ is the multiplicative order of $\lambda$, and \( N \) is the number of integers \( k \) such that 
	$
	0 \leq k < \left\lfloor \frac{m - 1}{r} \right\rfloor + 1 \quad \text{and} \quad \gcd(1 + rk, m) = 1.
	$
	This formula help us to reduce the number of cases to check whether such matrices is MDS. Moreover, we give complete characterization of $g$-circulant MDS matrices of order 3 and 4. Additionally, inspired by skew polynomial rings, we construct a new variant of $g$-circulant matrix. In the last, we provide some examples related to our findings.

	\end{abstract}
\noindent\textit{Keywords: }{Circulant matrices, $g$-circulant matrices, consta-$g$-circulant matrices, involutory, MDS matrices}   \\
\textit{2020 Mathematics Subject Classification: }{94A60, 12E20, 15A99, 15B33, 11T71} 
\section{Introduction}

The notion of confusion and diffusion with respect to design of encryption systems were first introduced by  Shannon in his seminal paper ``Communication Theory of Secrecy Systems" \cite{shannon1949communication}. The goal of the confusion layer is to hide the relationship between the key and the ciphertext, while the diffusion layer hides the relationship between the ciphertext and plaintext. When used in an iterated block cipher with layers of confusion and diffusion, these properties ensure that every message  and  secret-key bit have a non-linear impact on every bit of the ciphertext. The choice of diffusion layer influences both the security and efficiency of  cryptographic primitives. The diffusion layer plays an important role in providing security against differential cryptanalysis and linear cryptanalysis, as highlighted in \cite{daemen2002design}.  Among the attributes of a diffusion layer, the branch number \cite{daemen2002design} plays a critical role in estimating the diffusion property. Optimal diffusion can be achieved by multipermutation, as demonstrated in \cite{schnorr1994black, vaudenay1994need}, or using Maximum Distance Separable matrices \cite{daemen2002design}.

MDS matrices are widely used in many block ciphers and hash functions. For example, ciphers like AES \cite{daemen2002design} and SHARK \cite{rijmen1996cipher}, etc., and hash functions like PHOTON \cite{guo2011photon}  use MDS matrices in their diffusion layers. In many block ciphers, if $M$ is the MDS matrix used during encryption process, then $M^{-1}$ is used in decryption. Therefore, it is  important to use MDS matrices whose inverse is easy to compute and efficient to implement in software/hardware, unlike  Feistel based construction, which does not require inverse linear transformation. For the purpose of reducing implementation costs, research focuses on  circulant matrices or their variants and recursive MDS matrices. Now, we review some studies related to these two types of matrices.

Among various matrix structures, circulant matrices have garnered significant attention due to their computational efficiency and algebraic properties. A circulant matrix is a unique type of matrix where each row vector is rotated one element to the right in relation to the previous row vector. The circulant matrix of order $m$ has the advantageous of having a maximum of $m$ different components, so it is better to store $m$ elements rather than $m^2$ elements. This inherent structure makes circulant matrices highly desirable in applications where memory and computational efficiency are crucial. In particular, circulant MDS matrices have been widely studied for their role in the MixColumns transformation of the Advanced Encryption Standard (AES) \cite{daemen2002design}. Extensive research has been conducted on circulant MDS matrices over finite fields and finite rings, leading to significant theoretical advancements \cite{adhiguna2022orthogonal, ali2024circulant, ali2024xor, cauchois2019circulant, gupta2015cryptographically, kesarwani2021recursive}. Notably, Gupta et al. \cite{gupta2015cryptographically} proved in 2014 that neither orthogonal circulant MDS matrices of order $2^k$ nor involutory circulant MDS matrices of size $k$ (for $k \geq 3$) exist over finite fields of characteristic 2. Cauchois \cite{cauchois2019circulant} developed the algebraic framework of circulant matrices. However, in 2021, Kesarwani et al. \cite{ kesarwani2019exhaustive} performed an  exhaustive search for circulant MDS matrices over $GL(4, \mathbb{F}_2)$, specifically for  orders 4, 5, 6, 7, and 8. In \cite{kesarwani2021recursive}, Kesarwani et al.\ studied the construction of MDS matrices over finite rings.
Furthermore, Ali et al. \cite{ali2024circulant} demonstrated the nonexistence of circulant involutory MDS matrices over finite commutative rings of characteristic 2. These results highlight the inherent structural limitations of circulant matrices in cryptographic applications.

 In 1961, Friedman \cite{friedman1961eigenvalues} introduced $g$-circulant matrices as a generalization of the circulant matrix. After the several authors \cite{ablow1963roots,Davis1979Circulant, friedman1961eigenvalues} studied the $g$-circulant matrices and their properties. In \cite{friedman1961eigenvalues}, a representation of $g$-circulant matrices using permutation matrices is provided. Note that, in \cite{gupta2015cryptographically}, authors investigated MDS circulant matrices using representation of circulant matrix in terms of permutation matrix. Subsequently, in 2016, Liu and Sim proved that it is possible to construct left-circulant involutory matrices with the MDS property. They also introduced the notion of cyclic matrices and a connection between cyclic and circulant matrices. Importantly, circulant matrices are known for their efficient implementation. In \cite{liu2016lightweight}, both round-based and serial-based implementation strategies were proposed for circulant matrices. A similar implementation strategy extends naturally to cyclic matrices due to their relation to circulant ones. In serial-based implementation, we only need to store and apply the first row of the cyclic matrix. After computing the first output component, the input vector is updated by a corresponding permutation, and the same row is reapplied to compute the next component. This process is repeated until the full output vector is obtained. Thus, cyclic matrices, like circulant ones, support lightweight and efficient implementation using minimal resources.

Following this direction, the authors in \cite{liu2016lightweight} applied computational search methods to show that involutory cyclic MDS matrices do not exist for orders 4 and 8. In a recent study \cite{Chatterjee2024OnMDS}, Chatterjee and Laha\ investigated 
$g$-circulant MDS matrices. In related works \cite{Chatterjee2024ANoteonMDS2, Chatterjee2024ACharacterization}, the authors examined circulant matrices characterized by semi-involutory and semi-orthogonal properties. Furthermore, Chatterjee and Laha \cite{Chatterjee2024Anote} proved several non-existence results concerning 
$g$-circulant MDS matrices over finite fields

Motivated by the above recent developments of $g$-circulant matrices and the study of consta-cyclic codes~\cite{Bakshi2012Aclass}, this paper introduces a new class of matrices called consta-$g$-circulant matrices. These matrices generalize both circulant and $g$-circulant matrices by incorporating a fixed nonzero element $\lambda \in \mathbb{F}_q$ that alters the usual row-shifting mechanism. Each consta-$g$-circulant matrix is associated with a polynomial of the form
\[
h(x) = x^m - \lambda + \sum_{i=0}^{m-1} h_i x^i, \quad \text{where } h_0, h_1, \dots, h_{m-1} \in \mathbb{F}_q.
\]
We derive an upper bound on the number of such matrices and identify the conditions for their invertibility. These findings reduce the computational complexity involved in searching for suitable consta-$g$-circulant  matrices, particularly in cryptographic contexts. An explicit count of invertible matrices is also given, supplemented by a parameter table. We further study when these matrices are both \emph{involutory} and  MDS. 

Building upon the theory of skew polynomial rings, we introduce a more general class of matrices called consta-$\theta_g$-circulant matrices, which involve a field automorphism $\theta$. We analyze when these matrices also exhibit MDS and involutory properties. In addition, we explore \emph{semi-involutory matrices}, where the inverse takes the form $A^{-1} = DAD$, with $D$ being a fixed diagonal matrix over $\mathbb{F}_q$. This form offers practical benefits in terms of both structure and efficiency. We conclude the paper with several concrete examples.

The rest of the paper is organized as follows. Section~2 covers the background on circulant and \(g\)-circulant matrices over finite fields and explains their relevance in achieving strong diffusion. In Section~3, we introduce consta-\(g\)-circulant matrices, discuss their main properties, and give conditions for when they are invertible. We also provide a count of how many such matrices exist. In Section 4, we provide the complete characterization of consta-$g$-circulant MDS mattrices of order 3 and 4. Section~5 focuses on consta-\(\theta_g\)-circulant matrices, which are built using automorphisms, and presents related results. We end our discussion by providing representative examples to support the results.

\section{\textbf{Notations and Preliminaries}}
We use the following notations throughout the paper:
\begin{table}[h]
	\centering
	\begin{tabular}{c p{10cm}}
		\hline
		$n,~m,~l,~k,~d,~t$& positive integers.\\
		\hline
		$\mathbb{F}_{q}$ &  the finite field of characteristic $p$ with $q$ elements, where $q$ is a  power of $p$. \\
		\hline
		$\mathbb{F}^n_{q}$&  the linear space over $\mathbb{F}_{q}$ of dimension $n$.\\
		\hline
		$M_{m\times n}(\mathbb{F}_{q})$ & the set of all   matrices of order $m\times n$  over   $\mathbb{F}_{q}$. \\
		\hline
		$GL(m,\mathbb{F}_q)$&  the general linear group consisting of all $m \times m$ invertible matrices over $\mathbb{F}_q$. \\
		\hline
		ord($\lambda$) & the multiplicative order of $\lambda \textup{ in } \mathbb{F}^*_{q}$\\
		\hline
		$\mathbb{F}_q[x]$& the polynomial ring \\
		\hline
		$\mathbb{F}_q[X;\theta]$& the skew polynomial ring, where $\theta:\mathbb{F}_q\longrightarrow \mathbb{F}_q$ is an automorphism. \\
		\hline
		$Q(x)$ &a polynomial over $\mathbb{F}_{q}$\\
		\hline
		$Q\langle X \rangle$ & a skew polynomial over $\mathbb{F}_{q}$. \\
		\hline
		$\#$&cardinality of a set.\\
		\hline
		\hline
	\end{tabular}
\end{table}\\
 In this section, we discuss some definitions and mathematical preliminaries that are
important in our context.  Now, we start with the definition of MDS code.

\begin{definition}\cite[ Definition 1.5 ]{dougherty2017algebraic}
	Let $\mathcal{C}$ be a \([n, m, d]\) linear code over \( \mathbb{F}_q \). Then, $\mathcal{C}$ is \emph{MDS} if its minimum distance \( d \) satisfies the Singleton Bound:
	\[
	d = n - m + 1.
	\]
\end{definition}

\begin{lemma}\cite{macwilliams1977theory}
	An $[n,m,d]$ code $\mathcal{C}$ with the generator matrix $G = [I \mid A]$, where $A\in M_{m\times {n-m}}(\mathbb{F}_q)$, is called \textit{MDS} if and only if every $i \times i$ submatrix of $A$ is non-singular, for all $i = 1, 2, \ldots, \min(m,n-m)$.
\end{lemma}
In the design of symmetric encryption schemes, invertible matrices are typically required. Consequently, we focus on square MDS matrices of order \( m \). These matrices constitute the redundant part of the generator matrix in the systematic form of MDS codes with length \( 2m \) and dimension \( m \), i.e., the parameter of the code is $[2m,m,m+1]$.

\begin{definition}
	A square matrix $A\in GL(m,\mathbb{F}_q) $  is said to be \textit{MDS} if every square submatrix of $A$ is non-singular.
\end{definition}
MDS matrices with efficiently implementable inverses are useful because of the use of inverse matrices in the decryption layer of an SPN-based block cipher \cite{daemen2002design}. Thus, MDS matrices with either involutory or orthogonal property are efficient in the context of implementation. Here, $A^{-1}$ denotes the inverse of $A$, and $I_m$ denotes the identity matrix of order $m$.

\begin{definition}
	A square matrix $A\in GL(m,\mathbb{F}_q)$  is said to be \textit{involutory} if $A^2 = I_m$.
\end{definition}

In 2021, Cheon et al.\ \cite{cheon2021semi} defined semi-involutory matrices as a generalization of the involutory matrices. The definition of a semi-involutory matrix is as follows:

\begin{definition}
	A non-singular matrix $M$ of order $m$ is said to be \textit{semi-involutory} if there exist non-singular diagonal matrices $D_1$ and $D_2$ such that
	\[
	M^{-1} = D_1 M D_2.
	\]
\end{definition}

Circulant and \( g \)-circulant matrices are special classes of matrices in which each row is obtained by applying a fixed shift to the preceding row. In the context of symmetric cryptography, particularly in the AES MixColumns operation \cite[Chapter-4]{paar2010understanding}, the employed MDS matrix is a circulant matrix. One of the key advantages of circulant matrices lies in their compact representation, that is, they contain at most \( n \) distinct entries, making it more efficient to store only \( n \) entries instead of \( n^2 \).

We now present the formal definitions of circulant and \( g \)-circulant matrices:

\begin{definition}
	The square matrix of the form
	\[
	\begin{bmatrix}
		c_0 & c_1 & c_2 & \cdots & c_{m-1} \\
		c_{m-1} & c_0 & c_1 & \cdots & c_{m-2} \\
		\vdots & \vdots & \vdots & \ddots & \vdots \\
		c_1 & c_2 & c_3 & \cdots & c_0
	\end{bmatrix}_{m \times m}
	\]
	is said to be a \textit{circulant matrix} (or, right circulant matrix) and is denoted by $\text{circ}(c_0, c_1, c_2, \ldots, c_{m-1})$.

On the other hand, the square matrix of the form
\[
\begin{bmatrix}
	c_0 & c_1 & c_2 & \cdots & c_{m-1} \\
	c_1 & c_2 & c_3 & \cdots & c_0 \\
	\vdots & \vdots & \vdots & \ddots & \vdots \\
	c_{m-1} & c_0 & c_1 & \cdots & c_{m-2}
\end{bmatrix}_{m\times m}
\]
is said to be a \textit{left-circulant matrix} and is denoted by $\text{l-circ}(c_0, c_1, c_2, \ldots, c_{m-1})$.
\end{definition}
In 1961, Friedman \cite{friedman1961eigenvalues} introduced a generalization of the circulant matrix, termed the \textit{$g$-circulant matrix}. In this matrix, each row (except the first) is derived from the previous row by cyclically shifting the elements by $g$ columns to the right. The formal definition is provided below.

\begin{definition}
	A \textit{$g$-circulant matrix} of order $m \times m$ is a matrix of the form

	\[
	\begin{bmatrix}
		c_0 & c_1 & \cdots & c_{m-1} \\
		c_{m-g} & c_{m-g+1} & \cdots & c_{m-1-g} \\
		c_{m-2g} & c_{m-2g+1} & \cdots & c_{m-1-2g} \\
		\vdots & \vdots & \ddots & \vdots \\
		c_g & c_{g+1} & \cdots & c_{g-1}
	\end{bmatrix}_{m\times m},
	\]
	with all subscripts taken modulo $m$.
\end{definition}

\begin{remark}
	Let \( A = (a_{i,j}) \) be a \( g \)-circulant matrix of order \( m \), with first row given by \( (c_0, c_1, \ldots, c_{m-1}) \). Then, the entries of \( A \) obey the recurrence relation
	\[
	a_{i,j} = a_{i+1,\,j+g},
	\]
	where the indices are taken modulo \( m \). Furthermore, each entry of \( A \) can be expressed explicitly as
	\[
	a_{i,j} = c_{j - ig \, (\mathrm{mod}\, m)}.
	\]
	
	When \( g = 1 \), the matrix reduces to a standard circulant matrix. If \( g \equiv -1 \pmod{m} \), the structure corresponds to that of a left-circulant matrix. Several important properties and applications of \( g \)-circulant matrices can be found in \cite{ablow1963roots, Davis1979Circulant}.
\end{remark}

Now, we discuss some important lemmas which help us to prove the results in the subsequent section.
\begin{lemma}\label{6Theorem4}\cite[Chapter 6, Theorem 6.18]{wan2011finite}
	Let $\mathbb{F}_q$ be a finite field $q$ elements, where $q$ is a prime power. If $q=2^n$, every element of $\mathbb{F}_{2^n}$ is a square element. Moreover, every element of $\mathbb{F}_{2^n}$ has a unique square root in $\mathbb{F}_{2^n}.$
\end{lemma}
\begin{lemma}\cite[Lemma 2.5 ]{dougherty2017algebraic}\label{6Lemma002}
Let $a_1, a_2, \ldots, a_s$ be ideals of a commutative ring $R$ which are relatively prime in pairs. Then,
	\(
	R / (a_1 \cdot a_2 \cdots a_s) \cong R / (a_1) \times R / (a_2) \times \cdots \times R / (a_s),\) where $(a_1\cdot a_2\cdots a_s)$ represents the ideal generated by $a_1\cdot a_2\cdots a_s$, and $(a_i)$ represents the ideal generated by $a_i.$ 
\end{lemma}
This leads us to the following well known \textit{Chinese Remainder Theorem(CRT)}:
\begin{lemma}\cite[Theorem 2.6 ]{dougherty2017algebraic}\label{6CRT}
	Let $R$ be a finite commutative ring with maximal ideals $\mathfrak{m}_1, \ldots, \mathfrak{m}_s$, where the index of stability of $\mathfrak{m}_i$ is $e_i$. Then, the map
	\[
	\Psi : R \to \prod_{i=1}^{s} R / \mathfrak{m}_i^{e_i}, \quad \text{defined by} \quad \Psi(x) = (x + \mathfrak{m}_1^{e_1}, \ldots, x + \mathfrak{m}_s^{e_s}),
	\]
	is a ring isomorphism, where $e_1,~e_2,\dots,e_s$ are the positive integers.
	
\end{lemma}
\section{Construction of consta-\(g\)-circulant MDS matrices}
In this section, we present a characterization of consta-\(g\)-circulant matrices. By applying the Chinese Remainder Theorem, we derive a formula for counting the number of invertible matrices of this type. Focusing on involutory matrices, we determine the conditions under which such matrices are MDS.

\quad Let \( \lambda \in \mathbb{F}_q^{\ast} \), the multiplicative group of non-zero elements of \( \mathbb{F}_q \). A \(\lambda\)-constacyclic code \( C \) of length \( n \) over \( \mathbb{F}_q \) can be identified with an ideal in the quotient algebra \( \mathbb{F}_q[x]/(x^m - \lambda) \), where \( (x^m - \lambda) \) denotes the ideal generated by the polynomial \( x^m - \lambda \) in the polynomial ring \( \mathbb{F}_q[x] \). Consequently, the code \( C \) is generated by a factor of \( x^m - \lambda \), referred to as the \emph{generator polynomial} of the \(\lambda\)-constacyclic code.
The elements of \( \mathbb{F}_q[x]/(x^m - \lambda) \) are equivalence classes of polynomials modulo \( x^m - \lambda \). By the division algorithm, every element in \( \mathbb{F}_q[x]/(x^m - \lambda) \) has a unique representative of degree less than \( m \). Thus, the elements of \( \mathbb{F}_q[x]/(x^m - \lambda) \) can be described as
\[
\frac{\mathbb{F}_q[x]}{(x^m - \lambda)} = \left\{ a_0 + a_1 \bar{x} + a_2 \bar{x}^2 + \cdots + a_{m-1} \bar{x}^{m-1} \ \middle|\ a_i \in \mathbb{F}_q \right\},
\]
where \( \bar{x} = x + (x^m - \lambda) \) denotes the class of \( x \) in the quotient ring. This characterization shows that \( \mathbb{F}_q[x]/(x^m - \lambda) \) is an \( \mathbb{F}_q \)-vector space of dimension \( m \), with basis \( \{ 1, \bar{x}, \bar{x}^2, \ldots, \bar{x}^{m-1} \} \). We denote $\mathcal{U}(\mathbb{F}_q[x]/(x^m - \lambda) )$ as the set of all unit elements in the quotient ring $\mathbb{F}_q[x]/(x^m - \lambda)$

\begin{definition}
	The Hamming weight of a polynomial \(\bar{f}(x) = \sum_{i=0}^{m-1} a_i \bar{x}^i\) $\in 
	\frac{\mathbb{F}_q[x]}{(x^m - \lambda)}
	$ is defined as
	\[
	\mathrm{wt}(\bar{f}) := \#\{i \in \{0, 1, \ldots, m-1\} \mid a_i \ne 0\},
	\]
	i.e., the number of nonzero coefficients in the representative of \(f\).
\end{definition}
 Inspired by the framework of constacyclic codes as studied in \cite{Bakshi2012Aclass}, this section introduces the concept of consta-$g$-circulant matrices. Prior to developing the concept of consta-$g$-circulant matrix, we establish several foundational results.

\begin{lemma}\label{6Lemma1}
	Let $g$ be any positive integer. Then, for any $\lambda \in \mathbb{F}_q$, $x^m- \lambda$ $|$ $x^{mg}-\lambda$, if $1\equiv g\mod \textup{ord}(\lambda).$
\end{lemma}
\begin{proof}
	Suppose $1\equiv g~\mod  \textup{ord}(\lambda)$, this implies $ \textup{ord}(\lambda)$ divides $g-1$. This implies that 
	\begin{equation}\label{6Equation1}
		\lambda^g=\lambda.
	\end{equation}
	Equation (\ref{6Equation1}) yields,
	$$x^{mg}-\lambda=x^{mg}-\lambda^g=(x^m-\lambda)s(x),$$
	for some $s(x)\in \mathbb{F}_q[x].$ Thus, $x^m- \lambda$ $|$ $x^{mg}-\lambda$.
\end{proof}
In light of Lemma~\ref{6Lemma1}, we present an analogue of \cite[Proposition 2]{cauchois2019circulant}.
\begin{proposition}\label{6Proposition1}
	Let $h(x)=x^m-\lambda+\sum_{i=0}^{m-1}h_ix^i\in \mathbb{F}_q[x],$ where $\lambda\in\mathbb{F}^*_q$ such that $g\equiv1\mod \textup{ord}(\lambda)$. Then, the map 
	\begin{eqnarray*}
		\phi:\frac{\mathbb{F}_q[x]}{(x^m-\lambda)} &\longrightarrow& \frac{\mathbb{F}_q[x]}{(x^m-\lambda)}\\		\phi(\bar{Q}(x))&:=&\bar{Q}(x^g)\bar{h}(x),\textup{ for all $\bar{Q}(x)$ in \( \mathbb{F}_q[x]/(x^m - \lambda) \),}
	\end{eqnarray*}
	is a vector space linear transformation.
\end{proposition}
\begin{proof}
	As a consequence of Lemma \ref{6Lemma1}, we see that \( \phi \) is well-defined. Next, we show that \( \phi \) is a linear transformation.
	 For this, let \( \bar{Q_1}(x)=a_0+a_1\bar{x}+\cdots+a_{m-1}\bar{x}^{m-1},~ \bar{Q}_2(x)= b_0+b_1\bar{x}+\cdots+b_{m-1}\bar{x}^{m-1} \in \frac{\mathbb{F}_q[x]}{(x^m - \lambda)} \). Then, we compute
	\begin{eqnarray*}
		\phi(\bar{Q}_1(x)+\bar{Q}_2(x))&=&\phi((a_0+b_0)+(a_1+b_1)\bar{x}+\cdots+(a_{m-1}+b_{m-1})\bar{x}^{m-1})\\
		&=&((a_0+b_0)+(a_1+b_1)\bar{x}^g+\cdots+(a_{m-1}+b_{m-1})\bar{x}^{g(m-1)})\bar{h}(x)\\
		&=&(a_0+a_1\bar{x}^g+\cdots+a_{m-1}\bar{x}^{g(m-1)})\bar{h}(x)+(b_0+b_1\bar{x}^g+\cdots+b_{m-1}\bar{x}^{g(m-1)})\bar{h}(x)\\
		&=&\phi(\bar{Q}_1(x))+\phi(\bar{Q}_2(x)).
	\end{eqnarray*}
	For $\alpha\in \mathbb{F}_q$, we have
	\begin{eqnarray*}
		\phi(\alpha \bar{Q}_1(x))&=&\phi(a_0\alpha+a_1\alpha \bar{x}+\cdots+\alpha a_{m-1}\bar{x}^{m-1})\\
		&=&(a_0\alpha+a_1\alpha \bar{x}^g+\cdots+\alpha a_{m-1}\bar{x}^{g(m-1)})\bar{h}(x)\\
		&=&\alpha(a_0\alpha+a_1\alpha \bar{x}^g+\cdots+\alpha a_{m-1}\bar{x}^{g(m-1)})\bar{h}(x)\\
		&=&\alpha\phi(\bar{Q}_1(x)).
	\end{eqnarray*} 
	Thus, $\phi$ is a vector space linear transformation.
	\end{proof}
\begin{remark}
	Matrix associated with linear transformation in Proposition \ref{6Proposition1} with respect to the  basis $B=\{1,\bar{x},\bar{x}^2,\dots,\bar{x}^{m-1}\}$, is called consta-$g$-circulant matrix. 	We use the notation $\mathcal{C}_g^{(\lambda)}(h_0, h_1, \ldots, h_{m-1})$ to denote a consta-$g$-circulant matrix associated with the polynomial 
	$
	h(x) = h_0 + h_1 x + h_2 x^2 + \cdots + h_{m-1} x^{m-1} +x^m-\lambda \in \mathbb{F}_q[x].
	$  
	 We give following examples:
	\begin{itemize}
		\item[(i)] If we take $m=5$, $g=3$ in Proposition \ref{6Proposition1}, we obtain  \textit{consta-3-circulant matrix:}
		$$\mathcal{C}_3^{(\lambda)}(h_0, h_1,h_2,h_3,h_4)=\begin{bmatrix}
			h_0&h_1&h_2&h_3&h_4\\
			\lambda h_2&\lambda h_3&\lambda h_4&h_0&h_1\\
			\lambda^2h_4&\lambda h_0&\lambda h_1&\lambda h_2&\lambda h_3\\
			\lambda^2h_1&	\lambda^2h_2&	\lambda^2h_3&	\lambda^2h_4&	\lambda h_0\\
				\lambda^3h_3&	\lambda^3h_4&	\lambda^2h_0&	\lambda^2h_1&	\lambda^2h_2
		\end{bmatrix}_{5\times 5}.$$ 
		\item [(ii)] If we take $\lambda=1$, then Proposition \ref{6Proposition1} yields \textit{$g$-circulant matrix}, i.e., 
		$$\mathcal{C}_g^{(1)}(h_0, h_1, \ldots, h_{m-1})=\begin{bmatrix}
			h_0&h_1&\cdots &h_{m-1}\\
			h_{m-g}&h_{m-g+1}&\cdots&h\\
			h_{m-2g}&h_{m-2g+1}&\cdots&h_{m-1-2g}\\
			\vdots&\vdots&\ddots&\vdots\\
			h_g&h_{g+1}&\cdots&h_{g-1}
					\end{bmatrix}_{m \times m}.$$ 
						\item [(iii)] If we take $\lambda=1$ and $g=1$, then Proposition \ref{6Proposition1} yields \textit{circulant matrix}, i.e., 
					$$\mathcal{C}_1^{(1)}(h_0, h_1, \ldots, h_{m-1})=\begin{bmatrix}
						h_0&h_1&\cdots &h_{m-1}\\
						h_{m-1}&h_{m-2}&\cdots&h_{m-2}\\
						h_{m-2}&h_{m-1}&\cdots&h_{m-3}\\
						\vdots&\vdots&\ddots&\vdots\\
						h_1&h_{2}&\cdots&h_{0}
					\end{bmatrix}_{m\times m}.$$ 
						\item [(iv)] If we take $\lambda=1$ and $g=m-1$, then Proposition \ref{6Proposition1} yields \textit{left circulant matrix}, i.e., 
					$$\mathcal{C}_{m-1}^{(1)}(h_0, h_1, \ldots, h_{m-1})=\begin{bmatrix}
						h_0&h_1&\cdots &h_{m-1}\\
						h_{1}&h_{2}&\cdots&h_{0}\\
						h_{2}&h_{3}&\cdots&h_{1}\\
						\vdots&\vdots&\ddots&\vdots\\
						h_{m-1}&h_{0}&\cdots&h_{m-2}
					\end{bmatrix}_{m \times m}.$$ 
	\end{itemize}
\end{remark}

\begin{proposition}\label{6counting1}
The upper bound of the total number of consta-$g$-circulant matrix associated with the Proposition \ref{6Proposition1} is $\left( \left\lfloor \frac{m - 1}{\operatorname{ord}(\lambda)} \right\rfloor + 1 \right) \cdot q^m$, where  $\lfloor \cdot \rfloor$ represent the greatest integer function.
\end{proposition}\label{6Proposition002}
\begin{proof}
	Let us denote \( r = \operatorname{ord}(\lambda) \), where \( \lambda \in \mathbb{F}_q^* \), and let \( m \) be a positive integer. According to Proposition~\ref{6Proposition1}, for each integer \( g \in \{1, 2, \dots, m\} \), the map
	\[
	\phi: \frac{\mathbb{F}_q[x]}{(x^m - \lambda)} \to \frac{\mathbb{F}_q[x]}{(x^m - \lambda)}, \quad \phi(\bar{Q}(x)) := \bar{Q}(x^g)\bar{h}(x)
	\]
	is a vector space linear transformation, provided that
	\(
	g \equiv 1 \pmod{r}.
	\) We seek the number of integers \( g \in \{1, 2, \dots, m\} \) such that \( g \equiv 1 \mod r \). These are precisely the values in the arithmetic progression
	\[
	g = 1,\, 1 + r,\, 1 + 2r,\, \dotsc \leq m.
	\]
	Let \( N_g \) be the number of such values. This is the number of integers of the form \( g = 1 + kr \) for integers \( k \geq 0 \), such that
	\[
	1 + kr \leq m \quad \Rightarrow \quad k \leq \left\lfloor \frac{m - 1}{r} \right\rfloor.
	\]
	Thus,
	\[
	N_g = \left\lfloor \frac{m - 1}{r} \right\rfloor + 1.
	\]
	
\noindent 	Note that \( \bar{h}(x) \in \mathbb{F}_q[x]/(x^m - \lambda) \), which is a vector space of dimension \( m \) over \( \mathbb{F}_q \). Therefore, the number of possible choices for \( \bar{h}(x) \) is $q^m.$ Each admissible value of \( g \) yields \( q^m \) possible linear maps \( \phi \), corresponding to the choices of \( \bar{h}(x) \). Hence, the total number of such transformations is atmost
	$$ N_g \cdot q^m = \left( \left\lfloor \frac{m - 1}{r} \right\rfloor + 1 \right) \cdot q^m.
	$$
	Since each linear transformation \( \phi \) corresponds to a consta-$g$-circulant matrix, this is also the upper bound of  such matrices.
	
\end{proof}
\begin{proposition}
	Let $A=\mathcal{C}_3^{(\lambda)}(h_0, h_1,h_2,\dots, h_{m-1})$ and $B=\mathcal{C}_3^{(\lambda)}(h'_0, h'_1,h'_2,\dots, h'_{m-1}).$ Then, $AB$ is consta-$g_1\cdot g_2$-circulant matrix.
\end{proposition}
\begin{proof}
	The proof follows directly from Proposition \ref{6Proposition1}.
	
\end{proof}
\begin{lemma}\cite[Lemma 3.16]{Chatterjee2024Anote}\label{6Lemma2}
	Let $S=\{kg \mod~m;~k=0,1,2,\dots,m-1\}.$ Then, S will be complete residue system modulo $m$ $\iff$ $\gcd(m,g)=1$.
\end{lemma}
The above lemma ensures that exponentiation by \(g\) induces a permutation on the set of exponents modulo \(m\), when \(\gcd(m, g) = 1\). This property plays a key role in establishing the bijectivity of the map defined in the following proposition:

\begin{proposition}\label{6Proposition2}
	Let $\gcd(m,g)$=1 and $\gcd(h(x),~x^m-\lambda)=1$. Then, the map 
	\begin{eqnarray*}
	\phi:\frac{\mathbb{F}_q[x]}{(x^m-\lambda)} &\longrightarrow& \frac{\mathbb{F}_q[x]}{(x^m-\lambda)}\\		\phi(\bar{Q}(x))&:=&\bar{Q}(x^g)\bar{h}(x),\textup{ for all $\bar{Q}(x)$ in \( \mathbb{F}_q[x]/(x^m - \lambda) \),}
\end{eqnarray*}
	is bijective. 
\end{proposition}
\begin{proof}
	Let $\bar{Q}(x)=q_0+q_1\bar{x}+\cdots+q_{m-1}\bar{x}^{m-1}\in Ker~\phi$. Then, 
	\begin{eqnarray*}
		\phi(\bar{Q}(x))&=&\bar{0}.\\
	\end{eqnarray*}
	This implies that
	\begin{eqnarray*}	
			\bar{Q}(x^g)\bar{h}(x)=\bar{0}~
	\end{eqnarray*}
	This leads to $(x^m-\lambda)|(Q(x^g)h(x)).$ Since $\gcd(h(x),~x^m-\lambda)=1,$ so $x^m-\lambda |Q(x^g).$ It follows that $Q(x^g)=0 \mod(x^m-\lambda)$. Since $\gcd(m,g)=1,$ so by invoking of Lemma \ref{6Lemma2}, we obtain
	$q_0+q_1x^{i_1}+\cdots+q_{m-1}x^{i_{m-1}}=0 \mod (x^m-\lambda), \textup{ where}~\{0,i_1,i_2,\dots,i_{m-1}\}=\{0,1,2,\dots,m-1\}.$
Consequently, we have \( q_i = 0 \) for all \( i = 0,1,2,\dots,m-1 \), which implies that \( Q(x) = 0 \). Thus, \( \phi \) is injective. Moreover, since \( \frac{\mathbb{F}_q[x]}{(x^m-\lambda)} \) is finite, it follows that \( \phi \) is also surjective. Hence, \( \phi \) is bijective.

\end{proof}
In \cite[Theorem 4.1]{Chatterjee2024Anote}, Chatterjee and Laha proved that for any \( g \)-circulant matrix \( A \) of order \( m \), if \( \gcd(m, g) > 1 \), then \( A \) cannot be MDS. Advancing this result further, and specializing to the case \( \lambda = 1 \), we consider \( g \)-circulant matrices and establish the following theorem:

\begin{theorem}
			Let $A=\mathcal{C}_g^{(1)}(h_0, h_1,h_2,\dots, h_{m-1})$ be a matrix of order $m$. If gcd$(m,g)>1$ or gcd($h(x),x^m-1$) non-unit. Then, $A$ can not be an MDS matrix.  
\end{theorem}
\begin{proof}
		The proof is supported by the evidence presented in the Proposition \ref{6Proposition2}.
\end{proof}

The following theorem provides a formula for counting invertible consta-\( g \)-circulant matrices associated with the map defined in Proposition~\ref{6Proposition2}. This result is particularly valuable in cryptography, as it helps to significantly reduce the search space when constructing consta-$g$-circulant MDS matrices:

\begin{theorem}
	Let \( x^m - \lambda = \prod_{i=1}^{t} (f_i(x))^{e_i} \) be the factorization over \( \mathbb{F}_q \), where each \( f_i(x) \) is a monic irreducible polynomial and \( e_i \ge 1 \). Then, the total number of invertible consta-\( g \)-circulant matrices associated with the map in Proposition~\ref{6Proposition2} is
	\[
	N\cdot \prod_{i=1}^{t} \left( q^{\deg f_i} - 1 \right),
	\]
	where \( r = \operatorname{ord}(\lambda) \) and $N=\#\{k\in \mathbb{Z}: 0\leq k< \left\lfloor \frac{m - 1}{r} \right\rfloor +1,~\gcd(1+rk,m)=1\}$.
\end{theorem}

\begin{proof}
	The polynomial \( x^m - \lambda \) factorizes as 
$
x^m - \lambda = \prod_{i=1}^{t} (f_i(x))^{e_i},
$ where each \( f_i(x) \) is a monic irreducible polynomial and the factors are pairwise coprime. By the Lemma \ref{6CRT}, this factorization induces a ring isomorphism
\[
\frac{\mathbb{F}_q[x]}{( x^m - \lambda )} \cong \prod_{i=1}^{t} \frac{\mathbb{F}_q[x]}{( (f_i(x))^{e_i} )}.
\]
	An element in this ring is invertible if and only if its image in each component ring \( \frac{\mathbb{F}_q[x]}{( (f_i(x))^{e_i} )}  \) is invertible. Each of these component rings is a finite local ring, and its group of units is determined by the units in the corresponding residue field \( \frac{\mathbb{F}_q[x]}{( f_i(x) )} \). Therefore, the number of invertible elements in \( \frac{\mathbb{F}_q[x]}{( x^m - \lambda )} \) is
	\[
	\prod_{i=1}^{t} \left( q^{\deg f_i} - 1 \right).
	\]
	Let $N=\#\{k\in \mathbb{Z}: 0\leq k< \left\lfloor \frac{m - 1}{r} \right\rfloor +1,~\gcd(1+rk,m)=1\}.$
 Thus, multiplying the number of invertible polynomials (corresponding to invertible matrices) by the number of such distinct shift classes gives the total number of invertible consta-\( g \)-circulant matrices
	\[
	N \cdot \prod_{i=1}^{t} \left( q^{\deg f_i} - 1 \right).
	\]
	This completes the proof.
\end{proof}
\begin{example}
Consider $\mathbb{F}_{8},~\mathbb{F}_{9},~\mathbb{F}_{16},$ and $\mathbb{F}_{25}$ are the fields. Let $\omega$ be the  root of the polynomial $1+x+x^2$ in $\mathbb{F}_{16}$, $\alpha$ be the root of the polynomial $x^2+4x+1$ in $\mathbb{F}_{25}$ and $\beta$ be the root of the polynomial $x^2+1$ in $\mathbb{F}_9$. 
The following table presents counts of invertible consta-$g$-circulant matrices over finite fields $\mathbb{F}_q$ for various parameters. Each case includes the factorization of $x^m - \lambda$ and the corresponding number of invertible matrices:

\begin{table}[ht]
	\centering
	\caption{\textbf{\textit{ Invertible consta-$g$-circulant matrices for various parameters}}}
	\label{6tab:consta_g_circulant_10}
	\vspace{.3cm}
	
	\tiny{\begin{tabular}{|c|c|c|c|c|l|c|c|c|}
		\hline
		S. No.&\(m\) & \(q\) & \(\lambda\) & Order \(r\) & Factorization of \(x^m-\lambda\) over \(\mathbb{F}_q\) & \(N\) & \(\prod_i (q^{\deg f_i} -1)\) & Total invertible matrices\\&~&~&~&~&($x^m-\lambda=\prod_{i=0}^{t}f_{i
		
		}$)&~&~&($N\cdot\prod_{i=0}^{t}(q^{\deg f_i-1})$) \\ \hline
		
		1&2 & 8  & 1 & 1 & \(x^2 - 1 = (x-1)(x+1)\)                         & 1 & \((8^1-1)^2 = 7^2=49\)               & \(1 \times 49 = 49\) \\ \hline
		
		2&3 & 8  & 1 & 1 & \(x^3 - 1 = (x-1)(x^2 + x +1)\)                  & 2 & \((8^1-1)(8^2-1) = 7 \times 63 = 441\) & \(2 \times 441 = 882\) \\ \hline
		
		3&2 & 9  & 2 & 2 & \(x^2 - 2 = (x^2-2)\)                  & 1 & \((9^2-1)   = 80\) & \(1 \times 80 = 80\) \\ \hline
		
		4&4 & 8  & 1 & 1 & \(x^4 - 1 = (x-1)^4\) & 2 & \((8^1-1)^4 = 2401 \) & \(2 \times 2401 = 4802\) \\ \hline
		
		5&4 & 9 & 1 & 1 & \(x^4 - 1 = (x-1)(x+1)(x-\beta)(x-\beta^3)\)                         & 2 & $(9-1)(9-1)(9-1)(9-1)=4096$          & \(2 \times 4096 = 8192\) \\ \hline
		
		6&3 & 16 & 1 & 1 & \(x^3 - 1 = (x-1)(x - \omega)(x-\omega^2))\)                  & 2 & \((16^1-1)^3 = 3375\) & \(2 \times 3375 = 6750\) \\ \hline

		7&3 & 25 & 1 & 1 & \(x^3 - 1 = (x-1)(x-\alpha)(x-\alpha^2)\)                  & 2 & $(25-1)^3$=13824 & \(2 \times 13824 = 27648\) \\ \hline
		
		8&4 & 25 & 3 & 4 & \(x^4 - 1 = (x-1)(x+1)(x+2)(x+3)\)                  & 1 & \((25^1-1)^4=331776\) & \(1 \times 331776 = 331776\) \\ \hline

	\end{tabular}}
\end{table}

\end{example}
We now present a necessary and sufficient condition for a  matrix to be both MDS and involutory. The theorem below characterizes this condition in terms of the weight of a polynomials.

\begin{theorem}
	Let $g^2=m \cdot \textup{ord}(\lambda)+1$, and $h(x)=(x^m-\lambda)+\sum_{i=0}^{m-1}h_ix^i\in \mathbb{F}_q[x]$. Then, $\mathcal{C}_g^{(\lambda)}(h_0, h_1, \ldots, h_{m-1})$ is an involutory MDS matrix if and only if:
	\begin{enumerate}
		\item [(i)] for all $Q_1(x)\in \frac{\mathbb{F}_q[x]}{(x^m-\lambda)}$, we have 
		$$\textup{wt}(Q_1(x))+\textup{wt}(Q_1(x^g)h(x) \mod (x^m-\lambda))\geq m+1.$$
		\item [(ii)] $h(x)h(x^g)\equiv 1 \mod (x^m-\lambda).$
	\end{enumerate} 
\end{theorem}
\begin{proof}
	Let $\mathcal{C}_g^{(\lambda)}(h_0, h_1, \ldots, h_{m-1})$ is involutory MDS matrix. Then,  
	\begin{eqnarray*}
		&\Leftrightarrow& (I_m|\mathcal{C}_g^{(\lambda)}(h_0, h_1, \ldots, h_{m-1}))~\textup{is the generator matrix of MDS code}\\
		&\Leftrightarrow& \textup{for all } (q_0,q_1,\dots,q_{m-1})\in \mathbb{F}^m_{q},~\textup{wt}((q_0,q_1,\dots,q_{m-1})\cdot(I_m|\mathcal{C}_g^{(\lambda)}(h_0, h_1, \ldots, h_{m-1})))\geq m+1\\
		&\Leftrightarrow& \textup{wt}(q_0,q_1,\dots,q_{m-1})+\textup{wt}((q_0,q_1,\dots,q_{m-1})\cdot \mathcal{C}_g^{(\lambda)}(h_0, h_1, \ldots, h_{m-1}))\geq m+1.
	\end{eqnarray*}
	If one consider $Q_1( x)=\sum_{i=0}^{m-1}q_ix^i$, then 
	$$\textup{wt}(q_0,q_1,\dots,q_{m-1})=\textup{wt}(Q_1).$$
	From Proposition \ref{6Proposition1}, we know that $\mathcal{C}_g^{(\lambda)}(h_0, h_1, \ldots, h_{m-1})$ corresponds the multiplication by \( h (x) \) in \( \frac{\mathbb{F}_{q}[x]}{(x^m - \lambda)} \). Thus, we have
	$$\textup{wt}((q_0,q_1,\dots,q_{m-1})\cdot \mathcal{C}_g^{(\lambda)}(h_0, h_1, \ldots, h_{m-1}))=\textup{wt}(Q_1(x^g)h(x) \mod  (x^m-\lambda)).$$
	Since $\mathcal{C}_g^{(\lambda)}(h_0, h_1, \ldots, h_{m-1})$ is an involutory matrix, so we conclude  $\phi \circ \phi =I|_{\mathbb{F}_q}$, i.e., for all $\bar{Q}_1(x)\in \frac{\mathbb{F}_q[x]}{(x^m-\lambda)}$, we have
	\begin{eqnarray}\label{6Equation2}
		\bar{Q}_1(x)\notag&=&\phi\circ\phi(\bar{Q}_1(x))\\
		\notag&=&\phi(\bar{Q}_1(x^g)\bar{h}(x))\\
		&=&\bar{Q}_1(x^{g^2})\bar{h}(x^g)\bar{h}(x) .
	\end{eqnarray}
Since $g^2=m \cdot \textup{ord}(\lambda)+1$, 	Equation (\ref{6Equation2}), yields that   $h(x)h(x^g)\equiv 1 \mod (x^m-\lambda).$ 
\end{proof}
\section{Construction of 3\(\times\)3 and 4\(\times\)4 \(g\)-circulant MDS matrices}
This section presents the construction of \( g \)-circulant matrices of orders 3 and 4. Furthermore, an algorithm is provided for counting the number of such matrices.

\begin{theorem}\label{6Juarez}\cite[Theorem 5.2]{kesarwani4923673complexity}
Let $M$ be a matrix of order $m$.	If all the entries of \( M^{-1} \) are non-zero, then all the submatrices of order \( m - 1 \) of \( M \) are nonsingular.
\end{theorem}
In view of the above theorem, we have the following results:
\begin{theorem}\label{63gcirculant}
	Let $A=C^{(1)}_g(h_0,h_1,h_2)$ be a $g$-circulant matrix associated with a  polynomial $h(x)=h_0+h_1\bar{x}+h_2\bar{x}^2$ in  $\mathcal{U}\bigg(\frac{\mathbb{F}_q[x]}{(x^3-1)}\bigg)$. Then, $A$ is MDS if and only if $\textup{wt}(h(x))=\textup{wt}(h^{-1}(x))=3.$
\end{theorem}

\begin{theorem}\label{64gcirculant}
		Let $A=C^{(1)}_g(h_0,h_1,h_2,h_3)$ be a $g$-circulant matrix associated with polynomial $h(x)=h_0+h_1\bar{x}+h_2\bar{x}^2+h_3\bar{x}^4$ in $\mathcal{U}\bigg(\frac{\mathbb{F}_q[x]}{(x^4-1)}\bigg)$. If $\textup{wt}(h(x))=\textup{wt}(h^{-1}(x)),$ then $A$ is MDS if and only if every $2\times 2$ submatrix is non-singular.
\end{theorem}
The following algorithms are designed to facilitate the construction of $g$-circulant matrices satisfying the MDS property as characterized in the preceding theorems. Specifically, they aim to identify invertible polynomials of prescribed weight in the quotient rings $\mathbb{F}_q[x]/(x^3 - 1)$ and $\mathbb{F}_q[x]/(x^4 - 1)$, which serve as generating polynomials for the matrix entries. The first algorithm searches for polynomials $h(x)$ of Hamming weight $3$ such that both $h(x)$ and its inverse in $\mathbb{F}_q[x]/(x^3 - 1)$ have weight $3$, aligning with the condition in Theorem \ref{63gcirculant}. The second algorithm targets the case addressed in Theorem \ref{64gcirculant}, focusing on invertible polynomials in $\mathbb{F}_q[x]/(x^4 - 1)$ with a weight constraint that ensures the associated  $g$-circulant matrix is MDS, provided all $2 \times 2$ submatrices are non-singular. These algorithms significantly reduce the search space for constructing  $g$-circulant MDS matrices by filtering polynomials based on both invertibility and weight conditions.

\begin{algorithm}
	\caption{\small{:  Invertible polynomials of  Weight-3 in $\mathbb{F}_q[x]/(x^3 - 1)$}}
	\small{
		\begin{algorithmic}[1]\label{6Algo1}
			\Procedure{Find\_Weight3\_Invertible}{$q$}
			\State $F \gets \text{GF}(q)$
			\State $R_{\text{mod}} \gets F[x]/(x^3 - 1)$
			\State $x \gets$ generator of $R_{\text{mod}}$
			\For{$i = 1$ to $1000$}
			\State Choose random $a_0, a_1, a_2 \in F \setminus \{0\}$
			\State $h(x) \gets a_0 + a_1 x + a_2 x^2$ in $R_{\text{mod}}$
			\If{$h(x)$ is invertible in $R_{\text{mod}}$}
			\State $h^{-1}(x) \gets$ inverse of $h(x)$ in $R_{\text{mod}}$
			\If{HammingWeight($h^{-1}(x)$) $= 3$}
			\State \textbf{return} $h(x), h^{-1}(x)$
			\EndIf
			\EndIf
			\EndFor
			\State \textbf{return} `No valid $h(x)$ found'
			\EndProcedure
		\end{algorithmic}
	}
\end{algorithm}
\begin{example}
	 Consider $\mathbb{Z}_5$ be the finite field with 5 elements. Using Algorithm 1 implemented in SageMath, we found 8 distinct polynomials of the form \( h(x) = h_0 + h_1 x + h_2 x^2 \) in the unit group \( \mathcal{U}\bigg(\frac{\mathbb{F}_5[x]}{(x^3 - 1)}\bigg) \) such that both \( h(x) \) and its inverse \( h^{-1}(x) \) have weight 3. According to Theorem~\ref{63gcirculant}, the \(g\)-circulant matrices associated with these polynomials are MDS of order 3. The identified polynomials and their inverses are as follows:
	
	\begin{itemize}
		\item[(i)] \( h(x) = 2x^2 + x + 1,\quad h^{-1}(x) = x^2 + 2x + 1 \)
		\item [(ii)]\( h(x) = 3x^2 + 3x + 1,\quad h^{-1}(x) = 2x^2 + 2x + 4 \)
		\item[(iii)] \( h(x) = x^2 + x + 2,\quad h^{-1}(x) = x^2 + x + 2 \)
		\item [(iv)]\( h(x) = 4x^2 + 2x + 2,\quad h^{-1}(x) = 3x^2 + x + 3 \)
		\item[(v)] \( h(x) = 2x^2 + 4x + 2,\quad h^{-1}(x) = x^2 + 3x + 3 \)
		\item[(vi)] \( h(x) = 4x^2 + 4x + 3,\quad h^{-1}(x) = 4x^2 + 4x + 3 \)
		\item[(vii)] \( h(x) = 4x^2 + 3x + 4,\quad h^{-1}(x) = 3x^2 + 4x + 4.\)
	\end{itemize}
\end{example}

\begin{algorithm}
	\caption{: Invertible polynomials of  Weight-4 in $\mathbb{F}_q[x]/(x^4 - 1)$}
	\small{
		\begin{algorithmic}[1]
			\Procedure{Find\_Weight4\_Invertible}{$q$}
			\State $F \gets \text{GF}(q)$
			\State $R_{\text{mod}} \gets F[x]/(x^4 - 1)$
			\State $x \gets$ generator of $R_{\text{mod}}$
			\For{$i = 1$ to $1000$}
			\State Choose random $a_0, a_1, a_2, a_3 \in F \setminus \{0\}$
			\State $h(x) \gets a_0 + a_1 x + a_2 x^2 + a_3 x^3 \textup{ in } R_{\text{mod}}$
			\If{$h(x)$ is invertible}
			\State $h^{-1}(x) \gets$ inverse of $h(x)$ in $R_{\text{mod}}$
			\If{HammingWeight($h^{-1}(x)$) $= 4$}
			\State \textbf{return} $h(x), h^{-1}(x)$
			\EndIf
			\EndIf
			\EndFor
			\State \textbf{return} ``No valid $h(x)$ found''
			\EndProcedure
		\end{algorithmic}
	}
\end{algorithm}
\section{Variant of \(g\)-circulant MDS matrices involving skew polynomial ring}

	Let \(\theta:\mathbb{F}_{q} \longrightarrow \mathbb{F}_{q}\) be an automorphism. A skew polynomial ring over a field \(\mathbb{F}_{q}\) with automorphism \(\theta\), denoted \(\mathbb{F}_{q}[X; \theta]\), is defined as the set  
\[
\mathbb{F}_{q}[X; \theta] := \left\{ a_0 + a_1 X + \dots + a_{n} X^{n} : a_i \in \mathbb{F}_{q}, \, n \in \mathbb{N}\cup\{0\} \right\},
\]  
equipped with the usual addition of polynomials and a multiplication followed by the rule  
\begin{eqnarray*}
	X * a &=& \theta(a) X \quad \text{for all } a \in \mathbb{F}_{q},\\
	a*X&=& aX \quad \text{for all } a \in \mathbb{F}_{q}.
\end{eqnarray*}
Recall that an element $f\langle X \rangle \in \mathbb{F}_{q}[X; \theta] $ is called \emph{central} if $f\langle X \rangle\ast g\langle X \rangle = g\langle X \rangle\ast f\langle X \rangle$ for all $g\langle X \rangle \in \mathbb{F}_{q}[X; \theta] $, 
i.e., $f\langle X \rangle$ commutes with every element of $\mathbb{F}_{q}[X; \theta] $. The set of all elements of $\mathbb{F}_{q}[X; \theta] $ that commute with every element of $\mathbb{F}_{q}[X; \theta] $ is called the \emph{center} of $\mathbb{F}_{q}[X; \theta] $, which is defined and denoted by
\[
Z(\mathbb{F}_{q}[X; \theta] ) = \{ f\langle X \rangle \in \mathbb{F}_{q}[X; \theta]  : f\langle X \rangle\ast g\langle X \rangle = g\langle X \rangle\ast f\langle X \rangle \text{ for all } g\langle X \rangle \in \mathbb{F}_{q}[X; \theta]  \}.
\]
 We define the fixed field of \(\mathbb{F}_q\) under the mapping \(\theta\) as
\[
(\mathbb{F}_q)^\theta = \{a \in \mathbb{F}_q : \theta(a) = a\}.
\]

In the following proposition, we collect some properties of the skew polynomial ring $\mathbb{F}_q[X;\theta]$:

\begin{proposition}\label{6Proposition3} (\cite{ore1933theory} p. 483-486)
	Let $s$ be the order of the automorphism $\theta$, and denote by $(\mathbb{F}_q)^\theta$ the fixed subfield of $\mathbb{F}_q$ by $\theta$. Then,
	\begin{enumerate}
		\item [(i)]  The ring $\mathbb{F}_q[X;\theta]$ is, in general, a non-commutative ring unless $\theta$ is the identity automorphism of $\mathbb{F}_q$.
		\item [(ii)] An element $f\langle X \rangle \in \mathbb{F}_q[X;\theta]$ is central if and only if $f\langle X \rangle$ belongs to $(\mathbb{F}_q)^\theta[X^s]$, i.e., 
		\[ Z(\mathbb{F}_q[X;\theta]) = (\mathbb{F}_q)^\theta[X^s]. \]
		\item [(iii)]  Two-sided ideals of $\mathbb{F}_q[X;\theta]$ are generated by elements of the form $f\langle X^s\rangle\ast X^m$, where $m$ is an integer and $f\langle X \rangle \in \mathbb{F}_q[X;\theta]$.
		\item [(iv)] The ring $\mathbb{F}_q[X;\theta]$ is a right (resp. left) Euclidean domain.
	\end{enumerate}
\end{proposition} 
From the second and third parts of the Proposition~\ref{6Proposition3}, we can easily conclude that \(J = (X^m - \lambda)\) is a (two-sided) ideal of the skew polynomial ring \(\mathbb{F}_q[X;\theta]\), provided that \(\textup{ord}(\theta) \mid m\) and \(\lambda \in (\mathbb{F}_q)^\theta\). In this case, the quotient ring
\[
\frac{\mathbb{F}_q[X;\theta]}{(X^m - \lambda)}
\]
is well-defined. Each element of this quotient ring can be uniquely represented as a left coset of the form
\[
f\langle X \rangle + (X^m - \lambda),
\]
where \(f\langle X \rangle \in \mathbb{F}_q[X;\theta]\) is a skew polynomial of degree less than \(m\). That is, every element can be written as
\[
a_0 + a_1 \bar{X} + a_2 \bar{X}^2 + \cdots + a_{m-1} \bar{X}^{m-1},
\]
where \(a_i \in \mathbb{F}_q\), and \(\bar{X} = X + (X^m - \lambda)\) is the class of \(X\) in the quotient ring. Multiplication in this ring is induced by the skew multiplication in \(\mathbb{F}_q[X;\theta]\), which satisfies the rule
\[
X\ast a = \theta(a)X \quad \text{for all } a \in \mathbb{F}_q.
\]
\quad Now, we have following remarks:
\begin{remark}
If we consider a non-identity automorphism \(\theta: \mathbb{F}_q \rightarrow \mathbb{F}_q\), then the skew polynomial ring \(\mathbb{F}_q[X; \theta]\) is non-commutative. Moreover, \(\frac{\mathbb{F}_q[X; \theta]}{(X^m - \lambda)}\) is a left \(\mathbb{F}_q\)-vector space of dimension \(m\).
\end{remark}
\begin{remark}
Let \( f\langle X \rangle,~g\langle X \rangle \in \mathbb{F}_q[X;\theta] \). Then, 
$
f\langle X \rangle \equiv g\langle X \rangle \mod_\ast (X^m - \lambda)
$
if and only if there exists \( t\langle X \rangle \in \mathbb{F}_q[X;\theta] \) such that
$
f\langle X \rangle - g\langle X \rangle = t\langle X \rangle \ast (X^m - \lambda).
$
\end{remark}
In light of the preceding remark, we now present the following results:

\begin{proposition}\label{6Proposition4}
	Let $h\langle X \rangle=X^m-\lambda+\sum_{i=0}^{m-1}h_iX^i\in \mathbb{F}_q[X,\theta],$ where $\lambda\in ((\mathbb{F}_q)^\theta)^\ast$ such that $g\equiv1\mod \textup{ord}(\lambda)$. Then, the map 
	\begin{eqnarray*}
		\chi:\frac{\mathbb{F}_q[X,\theta]}{(X^m-\lambda)} \longrightarrow \frac{\mathbb{F}_q[X,\theta]}{(X^m-\lambda)}\\		\chi(\bar{Q}\langle X \rangle):=\bar{Q}\langle X^g\rangle *\bar{h}\langle X \rangle,
	\end{eqnarray*}
	is a vector space linear transformation.
\end{proposition}
\begin{proof}
	
	As a consequence of Lemma \ref{6Lemma1}, we see that \( \chi \) is well-defined. Next, we show that \( \chi \) is a linear transformation.
	For this, let \( \bar{Q}_1\langle X \rangle=a_0+a_1\bar{X}+\cdots+a_{m-1}\bar{X}^{m-1},~ \bar{Q}_2\langle X \rangle= b_0+b_1\bar{X}+\cdots+b_{m-1}\bar{X}^{m-1} \in \frac{\mathbb{F}_q[X,\theta]}{(X^m - \lambda)} \). Then, we compute
	\begin{eqnarray*}
		\chi(\bar{Q}_1\langle X \rangle+\bar{Q}_2\langle X \rangle)&=&\chi((a_0+b)+(a_1+b_1)\bar{X}+\cdots+(a_{m-1}+b_{m-1})\bar{X}^{m-1})\\
		&=&((a_0+b)+(a_1+b_1)\bar{X}^g+\cdots+(a_{m-1}+b_{m-1})\bar{X}^{g(m-1)})\ast \bar{h}\langle X \rangle\\
		&=&(a_0+a_1\bar{X}^g+\cdots+a_{m-1}\bar{X}^{g(m-1)})*\bar{h}\langle X \rangle+(b_0+b_1\bar{X}^g+\cdots+b_{m-1}\bar{X}^{g(m-1)})\\
		&~&\ast \bar{h}\langle X \rangle\\
		&=&\chi(\bar{Q}_1\langle X \rangle)+\chi(\bar{Q}_2\langle X \rangle).
	\end{eqnarray*}
	For $\alpha\in \mathbb{F}_q$, we have
	\begin{eqnarray*}
		\chi(\alpha \bar{Q}_1\langle X \rangle)&=&\chi(a_0\alpha+a_1\alpha \bar{X}+\cdots+\alpha a_{m-1}\bar{X}^{m-1})\\
		&=&(a_0\alpha+a_1\alpha \bar{X}^g+\cdots+\alpha a_{m-1}\bar{X}^{g(m-1)})\ast \bar{h}\langle X \rangle\\
		&=&\alpha(a_0\alpha+a_1\alpha \bar{X}^g+\cdots+\alpha a_{m-1}\bar{X}^{g(m-1)})\ast \bar{h}\langle X \rangle\\
		&=&\alpha\chi(\bar{Q}_1\langle X \rangle).
	\end{eqnarray*} 
	Thus, $\chi$ is a vector space linear transformation.
\end{proof}
	\begin{remark}
		Matrix associated with linear transformation in Proposition \ref{6Proposition4} with respect to the  basis $B=\{1,\bar{X},\bar{X}^2,\dots,\bar{X}^{m-1}\}$, is called consta-$\theta_g$-circulant matrix. We use the notation $\mathcal{C}_{\theta_g}^{(\lambda)}(h_0, h_1, \ldots, h_{m-1})$ to denote a consta-$\theta_g$-circulant matrix associated with the skew polynomial
	$
		h\langle X \rangle = h_0 + h_1 X + h_2 X^2 + \cdots + h_{m-1} X^{m-1} +X^m-\lambda\in \mathbb{F}_q[X; \theta]$
\end{remark}
Along similar lines as Proposition \ref{6counting1}, we establish the following result:
\begin{proposition}
The upper bound of the total number of consta-$\theta_g$-\textup{circulant} matrix associated with the Proposition \ref{6Proposition4} is $\left( \left\lfloor \frac{m - 1}{\operatorname{ord}(\lambda)} \right\rfloor + 1 \right) \cdot q^m$, where  $\lfloor \cdot \rfloor$ represent the greatest integer function function.
\end{proposition}

\begin{theorem}\label{6Theorem2}
	Let $g^2=m \cdot \textup{ord}(\lambda)+1$, and $h\langle X \rangle=(X^m-\lambda)+\sum_{i=0}^{m-1}h_iX^i\in \mathbb{F}_q[X;\theta]$. Then, $\mathcal{C}_{\theta_g}^{(\lambda)}(h_0, h_1, \ldots, h_{m-1})$ is an involutory MDS matrix if and only if:
	\begin{enumerate}
		\item [(i)] for all $\bar{Q}_1\langle X \rangle\in \frac{\mathbb{F}_q[X;~\theta]}{(X^m-\lambda)}$, we have 
		$$\textup{wt}(\bar{Q}_1\langle X \rangle)+\textup{wt}(\bar{Q}_1\langle X^g\rangle\ast \bar{h}\langle X \rangle \mod_\ast (X^m-\lambda))\geq m+1.$$
		\item [(ii)] $h\langle X \rangle\ast h\langle X^g\rangle \equiv 1 \mod_\ast (X^m-\lambda).$
	\end{enumerate} 
\end{theorem}
\begin{proof}
	Let $\mathcal{C}_{\theta_g}^{(\lambda)}(h_0, h_1, \ldots, h_{m-1})$ is involutory MDS matrix. Then,  
	\begin{eqnarray*}
		&\Leftrightarrow& (I_m|\mathcal{C}_{\theta_g}^{(\lambda)}(h_0, h_1, \ldots, h_{m-1}))~\textup{is the generator matrix of MDS code}\\
		&\Leftrightarrow& \textup{for all } (q_0,q_1,\dots,q_{m-1})\in \mathbb{F}^m_{q},~\textup{wt}((q_0,q_1,\dots,q_{m-1})\cdot(I_m|\mathcal{C}_{\theta_g}^{(\lambda)}(h_0, h_1, \ldots, h_{m-1})))\geq m+1\\
		&\Leftrightarrow& \textup{wt}(q_0,q_1,\dots,q_{m-1})+\textup{wt}((q_0,q_1,\dots,q_{m-1})\cdot \mathcal{C}_{\theta_g}^{(\lambda)}(h_0, h_1, \ldots, h_{m-1}))\geq m+1.
	\end{eqnarray*}
	If one consider $\bar{Q}_1\langle  X\rangle=\sum_{i=0}^{m-1}q_iX^i$, then 
	$$\textup{wt}(q_0,q_1,\dots,q_{m-1})=\textup{wt}(\bar{Q}_1).$$
	From Proposition \ref{6Proposition4}, we know that $\mathcal{C}_{\theta_g}^{(\lambda)}(h_0, h_1, \ldots, h_{m-1})$ corresponds the multiplication by \( h \langle X \rangle \) in \( \frac{\mathbb{F}_{q}[X;~ \theta]}{(X^m - 1)} \). Thus, we have:
	$$\textup{wt}((q_0,q_1,\dots,q_{m-1})\cdot \mathcal{C}_{\theta_g}^{(\lambda)}(h_0, h_1, \ldots, h_{m-1}))=\textup{wt}(\bar{Q}_1\langle X^g\rangle\ast \bar{h}\langle X \rangle \mod_\ast  (X^m-\lambda)).$$
	Since $\mathcal{C}_{\theta_g}^{(\lambda)}(h_0, h_1, \ldots, h_{m-1})$ is an involutory matrix, this implies that $\chi \circ \chi =I|_{\mathbb{F}_q}$, i.e., for all $\bar{Q}_1\langle X \rangle\in \frac{\mathbb{F}_q[X;~\theta]}{(X^m-\lambda)}$, we have
	\begin{eqnarray}\label{6Equation3}
		\bar{Q}\langle X \rangle\notag&=&\chi\circ\chi(\bar{Q}\langle X \rangle)\\
		\notag&=&\chi(\bar{Q}\langle X^g\rangle\ast \bar{h}\langle X \rangle)\\
		&=&\bar{Q}\langle X^{g^2}\rangle\ast \bar{h}\langle X^g\rangle \ast \bar{h}\langle X \rangle \mod_\ast (X^m-\lambda).
	\end{eqnarray}
	Since $g^2=m \cdot \textup{ord}(\lambda)+1$, 	Equation (\ref{6Equation3}), yields that   $h\langle X \rangle \ast h\langle X^g\rangle\equiv 1 \mod_\ast (X^m-\lambda).$ 
\end{proof}

 In view of Theorem \ref{6Theorem2}(ii), we have the following result:
\begin{proposition}\label{6Proposition 6}
	Let $A$ be a semi involutory matrix of order $m$ such that $A^{-1}=D_1AD_2$, where $D_1=Diag(c_1,c_1,\dots,c_{1})$ and $D_2=Diag(c_2,c_2,\dots,c_{2})$ for $c_1,c_2\in (\mathbb{F}_q)^\theta$. Then, $c_1c_2h\langle X \rangle\ast h\langle X^g\rangle\equiv 1 \mod_\ast (X^m-\lambda).$
\end{proposition}

\begin{definition}\cite[Definition 1.3]{render1995algebras}
	Let \( P\langle X \rangle\) and \( Q\langle X \rangle \) be two polynomials in $\mathbb{F}_q[X;~\theta]$ such that
	\[
	P  \langle X \rangle = \sum_{i=0}^{m} a_i X^i,
	\]
	\text{and}
	\[
	Q  \langle X \rangle = \sum_{i=0}^{n} b_i X^i.
	\]
	
	\noindent Then, the Hadamard product \( h\langle X \rangle\) of \( P \langle X \rangle\) and \( Q \langle X \rangle \) is given by
	\[
	h\langle X \rangle = P  \langle X \rangle \diamond Q  \langle X \rangle,
	\]
	where \( \diamond \) denotes the Hadamard product and the polynomial \( h\langle X \rangle \) is defined by its coefficients as follows:
	\[
	h\langle X \rangle = \sum_{i=0}^{\min(m,n)} c_i X^i,
	\]
	where
	\[
	c_i = a_i b_i \text{ for } i = 0, 1, \ldots, \min(m,n).
	\]

	
	\noindent The Hadamard \( s \)-power of \( P  \langle X \rangle \) is denoted by \( P  \langle X \rangle^{{\diamond} s} \), and  defined as
	\[
	P  \langle X \rangle^{{\diamond}s} = \sum_{i=0}^{m} (a_i)^s X^i.
	\]
	
\end{definition}

In the next theorem, we use the Hadamard product of polynomials to establish several equivalences related to the structural properties of \( g \)-circulant matrices. Specifically, we show that the MDS, involutory, and semi-involutory nature of a matrix is preserved under certain Hadamard powers of the associated polynomial.

\begin{theorem}\label{6Theorem 4}
	Let $\mathbb{F}_q$ be a finite field with $q$ elements, where $q$ is some power of 2. Let $A$ be a $\theta_g$-circulant matrix associated with the polynomial $P\langle X \rangle=\sum_{i=0}^{m-1}a_iX^i \in \mathbb{F}_{q}[X;~\theta]$. Then,
	\begin{enumerate}
		\item [(i)] $A$ is  MDS \textup{iff} matrix associated with the polynomial $P\langle X \rangle^{{\diamond} 2^s} $ for $s\in \mathbb{N}$ is MDS.
		\item [(ii)] $A$ is an involutory \textup{iff} matrix associated with the polynomial $P\langle X \rangle^{{\diamond} 2^s} $ for $s\in \mathbb{N}$ is an involutory matrix.
		\item [(iii)]  $A$ is a semi-involutory  \textup{iff} matrix associated with the polynomial $P\langle X \rangle^{{\diamond} 2^s} $ for $s\in \mathbb{N}$ is a semi- involutory matrix.
	\end{enumerate} 
\end{theorem}
	\begin{proof}[Proof of (i)]
	 Let $A$ is $g$-circulant MDS matrix associated with the polynomial $P\langle X \rangle=\sum_{i=0}^{m-1}a_{i}X^{i}.$ Then, we have 
	$$P( X ) =a^{2^s}_{0}+a^{2^s}_{1}X+\cdots+a^{2^s}_{m-1}X^{m-1}, ~a_i\in \mathbb{F}_q.$$
	In view of Proposition \ref{6Proposition4}, we have a $g$-circulant matrix of order $m$ associated with the polynomial $P\langle X \rangle^{{\diamond}2^s}, $ say $B$, i.e.,
	$$B=g-circ(a^{2^s}_0,a^{2^s}_{1},\dots,a^{2^s}_{m-1}).$$
	Suppose $B'$ be any submatrix of order $t$, i.e., 
	$$B'=\begin{bmatrix}
		b^{2^s}_{11}&b^{2^s}_{12}&\cdots&b^{2^s}_{1t}\\
		b^{2^s}_{21}&b^{2^s}_{22}&\cdots&b^{2^s}_{2t}\\
		\vdots&\vdots&\ddots&\vdots\\
		b^{2^s}_{t1}&b^{2^s}_{t2}&\cdots&b^{2^s}_{tt}\\
	\end{bmatrix},~where~A'=\begin{bmatrix}
		b_{11}&b_{12}&\cdots&b_{1t}\\
		b_{21}&b_{22}&\cdots&b_{2t}\\
		\vdots&\vdots&\ddots&\vdots\\
		b_{t1}&b_{t2}&\cdots&b_{tt}\\
	\end{bmatrix},$$
	is some submatrix of $A$ of order $t$. Since $A$ is an MDS matrix. This gives $\det(A')\neq0.$ Since $\det(B')=(\det(A'))^{2^s}$, so $\det(B')\neq 0$. Hence $B$ is an MDS matrix. In light of Lemma \ref{6Theorem4},  similarly we can prove the converse part of the theorem.
	\end{proof}
	\begin{proof}[Proof of (ii)]
	Let $A$ be an involutory $g$-circulant matrix of order $m$. Then, by Theorem \ref{6Theorem2}, we have 
	\begin{eqnarray}\label{6Eqn1}
		(a_0+a_1X+\cdots+a_{m-1}X^{m-1})\ast (a_0+a_1X^g+\cdots+a_{m-1}X^{g(m-1)})\nonumber&\equiv& 1 \mod_\ast (X^m-1). \\
		&~&
	\end{eqnarray}	
	Since we have a property, for any $a, ~b\in \mathbb{F}_{2^n}$, we have $(a+b)^{2^n}=a^{2^n}+b^{2^n}$. By using this property Equation \ref{6Eqn1} yields
	\begin{eqnarray}
		(a^{2^s}_0+a^{2^s}_1X+\cdots+a^{2^s}_{m-1}X^{m-1})\ast (a^{2^s}_0+a^{2^s}_1X^g+\cdots+a^{2^s}_{m-1}X^{g(m-1)})\nonumber&\equiv& 1 \mod_\ast (x^m-1).\\
		&~& 
	\end{eqnarray}
	This gives matrix associated with the polynomial $P(X )^{{\diamond} 2^s}$ is an involutory matrix.
	\end{proof}
	\begin{proof}[Proof of (iii)]
	The proof follows by an application of Proposition \ref{6Proposition 6}, using an approach similar to that of part (ii).
	
	\end{proof}
To illustrate the theorem in a simpler setting, we now consider the case when $\theta$ is the identity map. This leads to the following corollary:
	
	\begin{corollary}
	Let $\mathbb{F}_q$ be a finite field with $q$ elements, where $q$ is a power of $2$. Let $A$ be a $g$-circulant matrix associated with the polynomial $P(x) = \sum_{i=0}^{m-1} a_i X^i \in \mathbb{F}_q[x]$. Then,
		\begin{enumerate}
			\item[(i)] $A$ is MDS \textup{iff} the matrix associated with the polynomial $P(x)^{\diamond 2^s}$ for $s \in \mathbb{N}$ is MDS.
			\item[(ii)] $A$ is involutory \textup{iff} the matrix associated with the polynomial $P(x)^{\diamond 2^s}$ for $s \in \mathbb{N}$ is involutory.
			\item[(iii)] $A$ is semi-involutory \textup{iff} the matrix associated with the polynomial $P(x)^{\diamond 2^s}$ for $s \in \mathbb{N}$ is semi-involutory.
		\end{enumerate}
	\end{corollary}

Cauchois and Loidreau~\cite{cauchois2019circulant} established that if $\mathbb{F}_q$ is a finite field of characteristic $2$ and $m \geq 3$ is an odd integer, then there exists no involutory circulant MDS matrix of size $m$ over $\mathbb{F}_q$. In contrast to this limitation of classical circulant matrices, we demonstrate that generalized structures such as consta-$\theta_g$-circulant matrices are rich in providing involutory matrices over finite fields. 

As an illustration, the following examples present $\theta$-circulant matrices of orders $4$ and $3$ that are involutory, highlighting the potential of these generalized matrices in cryptographic applications.

\begin{example}
	Let $\mathbb{F}_{2^4}$ be  a finite field defined by the polynomial $X^4+X+1$ and $\beta$ a root of this polynomial. Let $\theta:\mathbb{F}_{2^4}\longrightarrow \mathbb{F}_{2^4}$ be the automorphisms defined as $\theta(a)=a^2$. Then, the matrix associated with the polynomial $h\langle X \rangle=X^4+1+(\beta^3+\beta^2)X^3+(\beta^2+\beta)X^2+\beta^2X+\beta$  is 
	$$\mathcal{C}_{\theta_1}^{(1)}(\beta, \beta^2, \beta^2+\beta, \beta^3+\beta^2)=\begin{bmatrix}
		\beta&\beta^2&\beta^2+\beta&\beta^3+\beta^2\\
		\beta^3+\beta^2+\beta+1&\beta^2&\beta+1&\beta^2+\beta+1\\
		\beta+\beta^2&\beta^3+\beta&\beta+1&\beta^2+1\\
		\beta&\beta^2+\beta+1&\beta^3&\beta^2+1\\
		
	\end{bmatrix}, $$
	which is an involutory and MDS matrix of order 4. Invoking  Theorem \ref{6Theorem 4}, we obtain three another involutory $\theta$-circulant matrix associated with the polynomials $h\langle X \rangle^{\circ^2}=X^4+1+(\beta^3+\beta^2+\beta+1)X^3+(\beta^2+\beta+1)X^2+(\beta+1)X+\beta^2+1$,  $h\langle X \rangle^{\circ^{2^2}}=X^4+1+(\beta^3+\beta)X^3+(\beta^2+\beta)X^2+(\beta^2+1)X+\beta$ and $h\langle X \rangle^{\circ^{2^3}}=X^4+1+\beta^3X^3+(\beta^2+\beta+1) X^2+\beta X+\beta^2$.
\end{example}
\begin{example}
	Let $\mathbb{F}_{2^3}$ be  a finite field defined by the polynomial $X^3+X+1$ and $\gamma$ a root of this polynomial. Let $\theta:\mathbb{F}_{2^3}\longrightarrow \mathbb{F}_{2^3}$ be the automorphisms defined as $\theta(a)=a^2$. The matrix associated with the polynomial $h\langle X \rangle=X^3+1+\gamma^2X^2+\gamma X+\gamma+1$  is 
	$$\mathcal{C}_{\theta_1}^{(1)}(\gamma+1, \gamma, \gamma^2)=\begin{bmatrix}
		\gamma+1&\gamma&\gamma\\
		\gamma^2&\gamma^2+1&\gamma^2\\
		\gamma^2+\gamma&\gamma^2+\gamma&\gamma^2+\gamma+1
	\end{bmatrix},$$
	which is an involutory and MDS matrix of order 3. Invoking  Theorem \ref{6Theorem 4}, we obtain two another involutory $\theta$-circulant matrix associated with the polynomials $h\langle X \rangle^{{\diamond}^2}=X^3+1+(\gamma^2+\gamma)X^2+\gamma^2X+\gamma^2+1$ and $h\langle X \rangle^{{\diamond}^{2^2}}=X^3+1+\gamma X^2+(\gamma^2+\gamma)X+\gamma^2+\gamma+1$.
\end{example}
\section{\textbf{Conclusion}}
\quad In this paper, we introduced and studied a new class of matrices called consta-\(g\)-circulant matrices, which generalize circulant and \(g\)-circulant matrices. We explored their algebraic structure, provided conditions for invertibility, and determined the total number of such matrices. Additionally, we examined cases where these matrices are involutory and MDS, and extended the framework to consta-\(\theta_g\)-circulant matrices by utilizing field automorphisms. We also provided algorithms for enumerating all 
$g$-circulant MDS matrices of orders 3 and 4. The study further considers semi-involutory matrices and presents theoretical results characterizing their structure and properties. As a direction for future work, the upper bound we established for the number of invertible consta-$g$-circulant matrices can be further refined to obtain tighter estimates. Moreover, it would be interesting to investigate similar structural and MDS related properties for orthogonal consta-$g$-circulant matrices, which may lead to new applications in areas such as cryptographic diffusion design and error correction.

\section{\bf Declarations}

 \noindent \textbf{Authors Contributions}\newline 
A. A. Khan prepared the original draft of the manuscript and contributed to its review and editing. The study was conceptualized and supervised by S. Ali, with visualization by B. Singh. All authors reviewed the manuscript.\newline
 
\noindent \textbf{Funding}\newline
The authors did not receive support from any organization for the submitted work.\newline

\noindent \textbf{Data Availability Statement}\newline
Data sharing is not applicable to this article as no data sets were
generated or analyzed during the current study.\newline

\noindent \textbf{Conflicts of Interest}\newline
The authors have no conflicts of interest to declare that are relevant to the content of this article. \newline

\end{document}